\documentclass[preprint,1pt]{elsarticle}

\usepackage{indentfirst}
\usepackage{multicol}
\usepackage{color}
\usepackage{mathrsfs}
\usepackage{amsthm}
\usepackage{latexsym}
\usepackage{amssymb}
\usepackage{amsmath}
\usepackage[right=0.7in,  top=0.8in,  bottom=1in,  left=0.7in]{geometry}

 \newtheorem{remark}{Remark} 
 
\newtheorem{theorem}{Theorem}[section]
\newtheorem{definition}{Definition}[section]
\newtheorem{lemma}{Lemma}[section]
\newtheorem{coro}{Corollary}[section]

\begin{document}

\begin{frontmatter}



\title{Holographic Algorithms on Bases of Rank 2}


\author[label1]{Zhiguo Fu
\footnote{Supported by Youth Foundation of Jilin University 450060445374. E-mail:fucomplex@hotmail.com}}
\address[label1]{Mathematics School of Jilin University Changchun 130024,  P.R. China}

\author[label2]{Fengqin Yang}
\address[label2]{School of Computer Science and Information Technology of Northeast Normal University Changchun 130117,  P.R. China}

\begin{abstract}
An essential problem in the design of holographic
algorithms
is to decide whether the required signatures
can be realized by matchgates under a suitable basis transformation (SRP).
For holographic algorithms on domain size 2,  \cite{string1, string3,  string7, string9} have built a systematical theory.
In this paper, we reduce SRP on domain size $k\geq 3$ to SRP on domain size 2 for holographic algorithms on bases of rank 2.
Furthermore, we generalize the collapse theorem of \cite{string6} to domain size $k\geq 3$.
\end{abstract}

\end{frontmatter}

\section{Introduction}
\label{Introduction}
L.  Valiant \cite{string23} introduced holographic algorithms with matchgates.  Computation in these algorithms is expressed and interpreted through a choice of linear basis vectors in an exponential ``holographic" mix.
Then the actual computation is carried out,
via the Holant Theorem,  by the Fisher-Kasteleyn-Temperley
algorithm for counting the number of perfect matchings
in a planar graph.
This methodology has produced polynomial time algorithms for a variety  of problems.  No polynomial time algorithms were known for any of these problems,  and some minor variations are known to be NP-hard.

For example,  Valiant showed that the restrictive SAT problem $\sharp_{7}$Pl-Rtw-Mon-3CNF (counting the number of satisfying assignments of a planar read-twice monotone 3CNF formula,  modulo 7) is solvable in P \cite{string24}.  The same counting problem $\sharp$Pl-Rtw-Mon-3CNF without mod 7 is known to be $\sharp$P-complete and the problem mod 2 is $\oplus$P-complete.  The surprising tractability mod 7 is due to the unexpected existence of some basis transformations for matchgate signatures.

For a general CSP-type counting problem,  one can assume there is a natural
parameter $k$,  called its domain size.
This is the range over which variables take values.
  For example,  Boolean CSP problems
all have domain size 2.  A $k$-coloring problem on graphs has domain size $k$.
In holographic algorithms one considers a linear transformation,
which can be expressed as a $2^{\ell}\times k$ matrix
$M = (\alpha_{1}, ~\alpha_{2},  \cdots,  \alpha_{k})$.
This is called a basis of $k$ components,
and $\ell$ is called the size of the basis\footnote{Following \cite{string23},
to allow greater flexibility in the design of holographic algorithms,
 a basis here may not be linearly independent,
e.g.,  when $\ell=1$, $k=3$.
However to be applicable to matchgates,  the number of rows must be a
power of 2.}.
A holographic algorithm is said to be on domain size $k$ if
the respective signatures are  realized by matchgates
using a basis of $k$ components.
When designing a holographic algorithm for any particular problem,  an
 essential step is to decide whether there is a linear basis for which certain signatures of both generators and recognizers can be simultaneously realized.
This is called Simultaneous Realizability Problem (SRP).

For SRP on domain size 2,
a systematic theory has been built in \cite{string1, string3,  string7, string9}. 
Recently, Valiant gave polynomial time algorithms for some interesting problems
by  holographic algorithms on $2\times 3$ bases in \cite{string25}, i.e. the domain size is 3. To understand the power of holographic algorithms, we need to
consider signatures on  domain size $k\geq 3$.
In the present paper, we give a method to reduce SRP on domain size $k\geq 3$ to SRP on domain size 2 if the signatures are realized on a basis of rank 2.

Obviously, utilizing bases of a higher size is always a theoretic possibility which may allow us to devise more holographic algorithms. But Cai and Lu proved a surprising result for holographic algorithms on domain size 2 in \cite{string6}: Any holographic algorithms on domain size 2 and a basis of size $\ell\geq 2$ which employs at least one non-degenerate generator can be simulated on a basis of size 1.
This is the collapse theorem for holographic algorithms on domain size 2.
In this paper, we give a collapse theorem for holographic algorithms on a $2^{\ell}\times k$ basis $M$, where $M$ has rank 2.

The above results are proved by ruling out a trivial case,  which happens when all the recognizers or generators are degenerate.
Holographic algorithms which only use degenerate recognizers or generators are trivial\cite{string6}.




\section{Background}
\label{Backgroun}
In this section,  we review some definitions and results.  More details can be found in \cite{string1, string3, string7, string23, string24}.

Let $G=(V, E, \omega)$
be a weighted undirected planar graph,
where $\omega$ assigns edge weights.
A generator (resp.  recognizer) matchgate $\Gamma$ is a tuple $(G, X)$
where $X\subseteq V$ is a set of external output (resp.  input) nodes.
The external nodes are ordered counter-clock wise on the external face.

Each matchgate is assigned a signature tensor.
 A generator $\Gamma$ with $n$ output nodes is assigned
a contravariant tensor $\textbf{G}$ of type $(\displaystyle _{0}^{n})$.
Under the standard basis $[e_{0}~e_{1}]
=
\begin{pmatrix}
1 & 0\\
0 & 1
\end{pmatrix}$,
 it takes the form $\underline{G}$
with $2^{n}$ entries,  where
\[
\underline{G}^{i_{1}i_{2}\cdots i_{n}}= {\rm PerfMatch}(G-Z),
i_{1}, i_{2}, \cdots, i_{n} \in \{0,  1\}.
\]
Here $Z$ is the subset of the output nodes having the characteristic
sequence $\chi_{Z}=i_{1}i_{2}\cdots i_{n}$,
$G-Z$ is the graph obtained from $G$ by removing $Z$ and its adjacent edges.
${\rm PerfMatch}(G-Z)$ is the sum,  over all perfect matchings $M$ of
$G-Z$,  of the product of the weights of matching edges in $M$.
(If all  weights are 1,  this is the number of perfect matchings.)
$\underline{G}$ is called the standard signature of the generator $\Gamma$.
We can view $\underline{G}$ as a column vector (whose entries are
ordered lexicographically according to $\chi_{Z}$).

Similarly a recognizer $\Gamma'=(G', X')$ with $n$ input nodes is assigned a covariant tensor $\textbf{R}$ of type $(\displaystyle _{n}^{0})$.
Under the standard basis,  it takes the form $\underline{R}$ with $2^{n}$ entries,
\[
\underline{R}_{i_{1}i_{2}\cdots i_{n}}= {\rm PerfMatch}(G'-Z),i_{1}, i_{2}, \cdots, i_{n} \in \{0,  1\},\]
where $Z$ is the subset of the input nodes having the characteristic sequence $\chi_{Z}=i_{1}i_{2}\cdots i_{n}$.
$\underline{R}$ is called the standard signature of the recognizer $\Gamma'$.
We can view $\underline{R}$ as a row vector (whose entries are
ordered lexicographically according to $\chi_{Z}$).

Generators and recognizers are essentially the same as far as
their standard signatures are concerned.
The distinction is how they transform with respect to a basis
transformation over some field (the default is $\bf C$).

 A $basis~M$ on domain size $k$ is a $2^{\ell}\times k$ matrix $(\alpha_{1}  , \alpha_{2} , \cdots , \alpha_{k})$ , where $\alpha_{i}$ has dimension $2^{\ell}$ (size $\ell$).
Under a basis $M$,  we can talk about the signature of a matchgate
after the transformation.


\begin{definition}
The contravariant tensor $\textbf{G}$ of a generator $\Gamma$ has signature $G$
(written as a column vector)
 under basis $M$ iff $M^{\otimes n}G=\underline{G}$ is the standard signature of the generator $\Gamma$.
\end{definition}


\begin{definition}
The covariant tensor $\textbf{R}$ of a recognizer $\Gamma'$ has signature $R$
(written as a row vector)
 under basis $M$ iff $\underline{R}M^{\otimes n}=R$ where $\underline{R}$ is the standard signature of the recognizer $\Gamma'$.
\end{definition}


\begin{definition}
A contravariant tensor $\textbf{G}$
(resp.  a covariant tensor $\textbf{R}$)
is realizable over a basis $M$ iff there exists a generator $\Gamma$ (resp.  a recognizer $\Gamma'$) such that $G$ (resp.  $R$) is the signature of $\Gamma$ (resp.  $\Gamma'$) under basis $M$.
They are simultaneously realizable if they are realizable over a
common basis.
\end{definition}





A matchgrid $\Omega=(A, B, C)$ is a weighted planar graph consisting of a disjoint union of: a set of $g$ generators $A=(A_{1}, A_{2}, \cdots, A_{g})$,  a set of $r$ recognizers $B=(B_{1}, B_{2}, \cdots, B_{r})$,  and a set of $f$ connecting edges $C=(C_{1}, C_{2}, \cdots, C_{f})$,  where each $C_{i}$ edge has weight 1 and joins an output node of a generator with an input node of a recognizer,  so that every input and output node in every constituent matchgate has exactly one such incident connecting edge.


Let $G(A_{i}, M)$ be the signature of generator $A_{i}$ under the basis $M$ and $R(B_{j}, M)$ be the signature of recognizer $B_{j}$ under the basis $M$.
Let $G=\bigotimes_{i=1}^{g}G(A_{i}, M)$ and $R=\bigotimes_{j=1}^{r}R(B_{j}, M)$
be their tensor product,
then $\rm{Holant}(\Omega)$ is defined to be the {\it contraction}
of these two product tensors
(the sum over all indices of the product of the
corresponding values of $G$ and $R$),
where the corresponding indices match up according to the $f$ connecting edges in $C$.

Valiant's {Holant} Theorem is

\begin{theorem}(Valiant \cite{string23})
For any mathcgrid $\Omega$ over any basis $M$,  let $\Gamma$ be its underlying weighted graph,
 then
\begin{center}
${\rm{Holant}}(\Omega)={\rm{PerfMatch}}(\Gamma)$.
\end{center}
\end{theorem}

The FKT algorithm can compute the weighted sum of
 perfect matchings $\rm{PerfMatch}(\Gamma)$   for a planar graph in P.  So $\rm{Holant}(\Omega)$ is computable in P.

 In the following discussion, we denote $\{1, 2, \cdots, k\}$
as $[k]$.

\section{Degenerate Recognizers}
\begin{definition}
A signature $R=(R_{i_{1}i_{2}\cdots i_{n}})$ (generator or recognizer) on domain size $k$ is degenerate iff
$R=v_{1}\otimes v_{2}\otimes\cdots\otimes v_{n}$, where $v_{i}$ are vectors of dimension $k$. Otherwise is non-degenerate.
\end{definition}

If all of the recognizers are degenerate, then the holographic algorithm is trivial. This is discussed in \cite{string6}.

\begin{definition}
For a recognizer $R=(R_{i_{1} i_{2} \cdots i_{n}})$ on domain size 2, i.e., $i_{j}\in\{1, 2\}$ for $1\leq j\leq n$, the $2\times 2^{n-1}$ matrix
\begin{equation*}
A_{R}(t)=\begin{pmatrix}
R_{i_{1}\cdots i_{t-1}1 i_{t+1}\cdots i_{n}}\\
R_{i_{1}\cdots i_{t-1}2 i_{t+1}\cdots i_{n}}
\end{pmatrix}
\end{equation*}
is called the $t$-th signature matrix of $R$
for $1\leq t\leq n$.
\end{definition}

\begin{lemma}\label{degenerate}
The recognizer $R$ on domain size 2 is degenerate iff rank$(A_{R}(t))\leq 1$ for $1\leq t\leq n$.
\end{lemma}
\begin{proof}
If $R$ is degenerate, it is obvious that rank$(A_{R}(t))\leq 1$ for $1\leq t\leq n$.

Conversely,
If there exists $t$ such that rank$(A_{R}(t))=0$, then $R$ is identical to zero and is degenerate.
Otherwise, rank$(A_{R}(t))=1$ for $1\leq t\leq n$.
 We will prove the Lemma by induction on the arity $n$.

For $n=2$, there exists a non-zero row $R_{i}$ in $A_{R}(1)$ and constants $a_{\sigma}$ such that $R_{\sigma}=a_{\sigma}R_{i}$ for $\sigma=1, 2$ since rank$(A_{R}(1))=1$.
Let $v_{1}=(a_{1}, a_{2}), v_{2}=R_{i}$, then
$R=v_{1}\otimes v_{2}$.

Inductively assume that the Theorem has  been proved for $\leq n-1$.
Since rank$(A_{R}(1))=1$, there exists a non-zero row $R_{i}$ in $A_{R}(1)$ and constants $a_{\sigma}$ such that $R_{\sigma}=a_{\sigma}R_{i}$ for $\sigma=1, 2$.
Note that $R_{i}$ is a signature of arity $n-1$ and all of its signature matrices are sub-matrices the signature matices of $A_{R}$.
By induction, $R_{i}$ is degenerate and there exist vectors
$v_{2}^{i}, v_{3}^{i},\cdots, v_{n}^{i}$ such that $R_{i}=v_{2}^{i}\otimes v_{3}^{i}\otimes\cdots\otimes v_{n}^{i}$.
Let $v_{1}=(a_{1}, a_{2})$ and $v_{t}=v_{t}^{i}$ for $2\leq t\leq n$, then $R=v_{1}\otimes v_{2}\otimes\cdots\otimes v_{n}$.
\end{proof}

In the present paper, assume that $M=(\alpha_{1}~~ \alpha_{2}~~ \cdots ~~\alpha_{k})$ is a $2^{\ell}\times k$ basis and of rank 2, where $k\geq 3$.
Then there exist $\sigma, \tau\in[k]$ such that the sub-matrix $(\alpha_{\sigma}~~\alpha_{\tau})$ of $M$ has rank 2.

\begin{definition}
For  $R=(R_{i_{1} i_{2} \cdots i_{n}})$ on domain size  $k$, i.e. $i_{t}\in[k]$,  the sub-signature
$R^{(s_{1}, s_{2}, \cdots, s_{d})}=(R_{j_{1} j_{2} \cdots j_{n}})$,
where $j_{1}, j_{2}, \cdots, j_{n}\in\{s_{1}, s_{2}, \cdots, s_{d}\}\subset[k]$,
is called the restriction of $R$ to  $\{s_{1}, s_{2}, \cdots, s_{d}\}$.
\end{definition}

\begin{lemma}\label{degenerate for the restriction}
If the recognizer $R$ on domain size $k$ is realizable on $M$ and non-degenerate, then $R^{(\sigma, \tau)}$ is non-degenerate.
\end{lemma}
\begin{proof}
Assume that $R=\underline{R}M^{\otimes n}$, where $\underline{R}$ is a standard signature.
Note that $R^{(\sigma, \tau)}=\underline{R}(\alpha_{\sigma}, \alpha_{\tau})^{\otimes n}$. Since $(\alpha_{\sigma}, \alpha_{\tau})$ has rank 2,
there exists a $2\times 2^{\ell}$ matrix $N$ such that $N(\alpha_{\sigma}, \alpha_{\tau})=I_{2}$. Thus
\begin{equation*}
R^{(\sigma, \tau)}N^{\otimes n}M^{\otimes n}=\underline{R}M^{\otimes n}=R.
\end{equation*}
If $R^{(\sigma, \tau)}$ is degenerate, i.e. $R^{(\sigma, \tau)}=v_{1}\otimes v_{2}\otimes\cdots\otimes v_{n}$, then
\begin{equation*}
R=v'_{1}\otimes v'_{2}\otimes\cdots\otimes v'_{n},
\end{equation*}
where $v'_{j}=v_{j}NM$ for $1\leq j\leq n$. This contradicts that $R$ is non-degenerate.
\end{proof}


\section{Holographic Algorithms on Bases of Rank 2}
For a recognizer $R=(R_{i_{1} i_{2} \cdots i_{n}})$ on domain size $k$,
the $nk^{(n-1)}\times 2$ matrix 
\begin{equation*}
 A_{\sigma,\tau}=\left(\begin{array}{ccccccccccccccccc}
R_{\sigma 1 \cdots 1 1}&R_{\sigma 1 \cdots 1 2}&\cdots &R_{\sigma k\cdots k k}
&R_{1 \sigma \cdots 1 1}&R_{1 \sigma \cdots 1 2}&\cdots &R_{k\sigma \cdots k k}&\cdots&
R_{1 1 \cdots 1 \sigma}&R_{1 1 \cdots 2 \sigma}&\cdots &R_{k k\cdots k \sigma}\\
R_{\tau 1 \cdots 1 1}&R_{\tau 1 \cdots 1 2}&\cdots &R_{\tau k\cdots k k}
&R_{1 \tau \cdots 1 1}&R_{1 \tau \cdots 1 2}&\cdots &R_{k\tau \cdots k k}&\cdots&
R_{1 1 \cdots 1 \tau}&R_{1 1 \cdots 2 \tau}&\cdots &R_{k k\cdots k \tau}
\end{array}\right)^{\texttt{t}}.
\end{equation*}
and the vector 
\begin{equation*}
b_{w}=(R_{w 1 \cdots 1 1}, R_{w 1 \cdots 1 2}, \cdots, R_{w k\cdots k k},
 R_{1 w \cdots 1 1}, R_{1 w \cdots 1 2}, \cdots, R_{k w\cdots k k},
\cdots,
R_{1 1 \cdots 1 w}, R_{1  1\cdots 2 w}, \cdots, R_{k k\cdots k w})^{\texttt{t}}
\end{equation*}
 of dimension $nk^{(n-1)}$
for $1\leq w\leq k$
(where the superscript $\texttt{t}$ denotes transpose). Then we have linear equations $A_{\sigma, \tau}X=b_{w}$  for $1\leq w\leq k$ from $R$.

\begin{lemma}\label{lemma4.1}
If $R$ is non-degenerate and realizable on $M$, then $A_{\sigma, \tau}$ has rank 2.
\end{lemma}
\begin{proof}
From Lemma \ref{degenerate for the restriction}, $R^{(\sigma, \tau)}$ is non-degenerate. Note that all of
the signature matrices of $R^{(\sigma, \tau)}$ are sub-matrices of $A_{\sigma, \tau}$, so $A_{\sigma, \tau}$ has rank 2 by Lemma \ref{degenerate}.
\end{proof}

By linear algebra, we have the following Lemma.

\begin{lemma}\label{lemma 4.2}
For the basis $M=(\alpha_{1}, \alpha_{2}, \cdots, \alpha_{k})$ of rank 2, where the sub-matrix $(\alpha_{\sigma}, \alpha_{\tau})$
has rank 2, there uniquely exists a $2\times k$ matrix
\[X_{\sigma, \tau}=\begin{pmatrix}
x^{\sigma}_{1}&x^{\sigma}_{2}&\cdots&x^{\sigma}_{k}\\
x^{\tau}_{1}&x^{\tau}_{2}&\cdots&x^{\tau}_{k}
\end{pmatrix}\]
such that $(\alpha_{\sigma}, \alpha_{\tau})X_{\sigma, \tau}=M$, where $X_{\sigma, \tau}$ has rank 2.
\end{lemma}

Note that $\alpha_{w}=x_{w}^{\sigma}\alpha_{\sigma}+x^{\tau}_{w}\alpha_{\tau}$ for $1\leq w\leq k$.

\begin{lemma}
\label{core}
If $R$ is non-degenerate and realizable on $M$, i.e., $R=\underline{R}M^{\otimes n}$, where $\underline{R}$ is a standard signature, then
$A_{\sigma, \tau}X=b_{w}$ has the unique solution $X_{w}=\begin{pmatrix}x^{\sigma}_{w}\\x^{\tau}_{w}\end{pmatrix}$ for $1\leq w\leq k$.
\end{lemma}
\begin{proof}
Firstly, the solution of $A_{\sigma, \tau}X=b_{w}$  is unique since rank$(A_{\sigma, \tau})=2$ from Lemma \ref{lemma4.1}.

Secondly,
$R_{i_{1} \cdots i_{t-1}wi_{t+1} \cdots i_{n}}=\langle \underline{R},
\alpha_{i_{1}}\otimes\cdots\otimes\alpha_{i_{t-1}}\otimes\alpha_{w}\otimes\alpha_{i_{t+1}}\otimes\cdots\otimes\alpha_{i_{n}}\rangle$, where
$i_{j}\in[k]$ for $j\neq t$ and $\langle\cdot,\cdot\rangle$ denotes inner products. Then
\begin{equation*}
\begin{split}
&R_{i_{1} \cdots i_{t-1}wi_{t+1} \cdots i_{n}}\\
&=\langle \underline{R},\alpha_{i_{1}}\otimes\cdots\otimes\alpha_{i_{t-1}}\otimes\alpha_{w}\otimes\alpha_{i_{t+1}}\otimes\cdots \otimes  \alpha_{i_{n}} \rangle\\
&=\langle \underline{R},\alpha_{i_{1}}\otimes\cdots\otimes\alpha_{i_{t-1}}\otimes(x_{w}^{\sigma}\alpha_{\sigma}+x_{w}^{\tau}\alpha_{\tau})\otimes\alpha_{i_{t+1}}\otimes\cdots\otimes\alpha_{i_{n}}\rangle\\
&=x_{w}^{\sigma}\langle \underline{R},\alpha_{i_{1}}\otimes\cdots\otimes\alpha_{i_{t-1}}\otimes\alpha_{\sigma}
\otimes\alpha_{i_{t+1}}\otimes\cdots\otimes\alpha_{i_{n}}\rangle+x_{w}^{\tau}\langle \underline{R},\alpha_{i_{1}}\otimes\cdots\otimes\alpha_{i_{t-1}}\otimes\alpha_{\tau}
\otimes\alpha_{i_{t+1}}\otimes\cdots\otimes\alpha_{i_{n}}\rangle\\
&=x_{w}^{\sigma}R_{i_{1} \cdots i_{t-1}\sigma i_{t+1} \cdots i_{n}}+x_{w}^{\tau}R_{i_{1} \cdots i_{t-1}\tau i_{t+1} \cdots i_{n}}.\\
\end{split}
\end{equation*}
This implies that $\begin{pmatrix}
x^{\sigma}_{w}\\x^{\tau}_{w}
\end{pmatrix}$ is a solution of $A_{\sigma, \tau}X=b_{w}$ and completes the proof.
\end{proof}

Lemma \ref{core} implies that if the recognizer $R$ is non-degenerate and realizable on some basis $M$ of rank 2, we can find the
matrix $X_{\sigma, \tau}$ from $R$.



\begin{theorem}
\label{theorem3.2}
The recognizers $R_{1},R_{2},\cdots,R_{r}$, where
$R_{1}$ is non-degenerate, are simultaneously realizable on some basis $M$  of rank 2 iff
the following conditions are satisfied:
\begin{itemize}
\item There exist $\sigma, \tau\in[k]$ such that $R_{1}^{(\sigma, \tau)}$ is non-degenerate.
\item The linear equations $A_{\sigma, \tau}X=b_{w}$ from $R_{1}$ has the unique solution
$\begin{pmatrix}
x^\sigma_w\\
x^\tau_w
\end{pmatrix}$,
for $1\leq w\leq k$ and $(\alpha_{\sigma}, \alpha_{\tau})X_{\sigma, \tau}=M$, where $X_{\sigma, \tau}=
\begin{pmatrix}
x^\sigma_1&x^\sigma_{2}&\cdots&x^{\sigma}_{k}\\
x^\tau_1&x^\tau_{2}&\cdots&x^\tau_{k}
\end{pmatrix}$.
\item There exists a $2^{\ell}\times 2$ basis $M_{(2)}$ such that $R_{i}^{(\sigma,\tau)}$, the restriction of $R_{i}$ to $\sigma,\tau$, are simultaneously realizable on $M_{(2)}$ for $1\leq i\leq r$.
\item $R_{i}=R_{i}^{(\sigma,\tau)}\displaystyle X^{\otimes n}_{\sigma,\tau}$ for $1\leq i\leq r$.

\end{itemize}
\end{theorem}
\begin{proof}
If $R_{1},R_{2},\cdots,R_{r}$ are simultaneously realizable on $M$,
then $R_{1}^{(\sigma,\tau)}$ is non-degenerate from Lemma \ref{degenerate for the restriction}.
From Lemma \ref{lemma 4.2} and Lemma \ref{core}, $A_{\sigma, \tau}X=b_{w}$ has the unique solution
$\begin{pmatrix}
x^\sigma_w\\
x^\tau_w
\end{pmatrix}$
and $(\alpha_{\sigma}, \alpha_{\tau})X_{\sigma, \tau}=M$. Let $M_{(2)}=(\alpha_{\sigma}, \alpha_{\tau})$.
Note that $\underline{R}_{i}(\alpha_{\sigma}, \alpha_{\tau})^{\otimes n}=R_{i}^{(\sigma, \tau)}$ if
$\underline{R}_{i}M^{\otimes n}=R_{i}$, so $R_i^{(\sigma, \tau)}$ are simultaneously realized on $M_{(2)}$
for $1\leq i\leq r$.
Furthermore, $R_{i}=\underline{R}_{i}M^{\otimes n}=\underline{R}_{i}(\alpha_{\sigma}, \alpha_{\tau})^{\otimes n}X_{\sigma, \tau}^{\otimes n}=
R_{i}^{(\sigma, \tau)}X_{\sigma, \tau}^{\otimes n}$ for $1\leq i\leq r$.



Conversely, since $R_{i}^{(\sigma,\tau)}$ are simultaneously realizable by some $2^\ell \times 2$ basis for some $\ell$, there exists a common basis $M_{(2)}=(\gamma_{1}~\gamma_{2})$ such that
\[R_{i}^{(\sigma,\tau)}=\underline{R}_{i}(M_{(2)})^{\otimes n}, 1\leq i\leq r.
\]
So
\[
R_{i}=R_{i}^{(\sigma,\tau)}X_{\sigma, \tau}^{\otimes n}
=\underline{R}_{i}M_{(2)}^{\otimes n}
X_{\sigma, \tau}^{\otimes n}=\underline{R}_{i}(M_{(2)}
X_{\sigma, \tau})^{\otimes n}, 1\leq i\leq r.
\]
This implies that $R_{1},R_{2},\cdots,R_{r}$ are simultaneously realizable on $M=M_{(2)}
X_{\sigma, \tau}$.
\end{proof}

\begin{remark}
Since the $2\times k$ matrix $X_{\sigma, \tau}$ has rank 2, there exists a $k\times k$
invertible
 matrix $X'_{\sigma, \tau}$ such that $X_{\sigma, \tau}X'_{\sigma, \tau}=
 \begin{pmatrix}
1&0&0&\cdots&0\\
0&1&0&\cdots&0
\end{pmatrix}$ by linear algebra.
\end{remark}


\begin{lemma}
\label{lemma3.3}
Let $R=\underline{R}M^{\otimes n}$, $R'=R(X'_{\sigma, \tau})^{\otimes n}$ and $\check{R}$ be the
restriction of $R'$ to  $\{1,2\}$,
 then $\check{R}=R^{(\sigma,\tau)}$.
\end{lemma}


\begin{proof}
\begin{equation*}\label{equation-for-R_0}
\begin{split}
R'&=\underline{R} M^{\otimes n}(X'_{\sigma, \tau})^{\otimes n}\\
&=\underline{R} (\alpha_\sigma, \alpha_\tau)^{\otimes n} X_{\sigma, \tau}^{\otimes n}(X'_{\sigma, \tau})^{\otimes n}\\
&=
\underline{R}
\begin{pmatrix}
\alpha_{\sigma}&\alpha_{\tau}&0&\cdots&0
\end{pmatrix}^{\otimes n}.
\end{split}
\end{equation*}
So $\check{R}=R^{(\sigma,\tau)}$.
\end{proof}


\vspace{.1in}

Now let $R_{1},R_{2},\cdots,R_{r}$  and $G_{1},G_{2},\cdots,G_{g}$ be recognizers and generators in a holographic algorithm.
Without loss of generality, assume that $R_{1}$ is non-degenerate and realizable on $M$, then we can find $X_{\sigma, \tau}$, $X'_{\sigma, \tau}$
from $R_{1}$.
Let $R_{i}'=R_{i}\displaystyle (X'_{\sigma, \tau})^{\otimes n}$ for $1\leq i\leq r$, $G_{j}'=\displaystyle (X'^{-1}_{\sigma, \tau})^{\otimes n}G_{j}$
for $1\leq j\leq g$. Note that $R'_i=\underline{R}_{i}  (\alpha_\sigma, \alpha_\tau, 0, 0, ..., 0)^{\otimes n}$ from the proof of Lemma \ref{lemma3.3}. And
let $\check{R}_{i}, \check{G}_{j}$ be the restriction of $R_{i}', G_{j}'$ to  $\{1, 2\}$  respectively,  then we have the following theorem.


\begin{theorem}
\label{theorem3.3}
Assume that $R_{1},R_{2},\cdots,R_{r}$ are simultaneously realizable on a basis of rank 2 and $R_{1}$ is non-degenerate, then we have
\begin{itemize}
\item If $R_{1},R_{2},\cdots,R_{r}$ and $G_{1},G_{2},\cdots,G_{g}$ are simultaneously realizable on a $2^{\ell}\times k$ basis $M=(\alpha_{1}, \alpha_{2}, \cdots, \alpha_{k})$ of rank 2, then  $\check{R}_{1},\check{R}_{2},\cdots,\check{R}_{r}$ and $\check{G}_{1},\check{G}_{2},\cdots,\check{G}_{g}$ are simultaneously realizable on the sub-matrix $M_{(2)}=(\alpha_{\sigma}, \alpha_{\tau})$ of $M$, where $M_{(2)}$ has rank 2.
\item If $\check{R}_{1},\check{R}_{2},\cdots,\check{R}_{r}$ and $\check{G}_{1},\check{G}_{2},\cdots,\check{G}_{g}$ are simultaneously realizable on a $2^{\ell}\times 2$ basis $M_{(2)}=(\gamma_{1}, \gamma_{2})$ of rank 2, then we can get $X_{\sigma, \tau}$ from $R_{1}$ such that $R_{1},R_{2},\cdots,R_{r}$ and $G_{1},G_{2},\cdots,G_{g}$ are simultaneously realizable on the $2^{\ell}\times k$ basis $M_{(2)}X_{\sigma, \tau}$.
\end{itemize}
i.e., $R_{1},R_{2},\cdots,R_{r}$ and $G_{1},G_{2},\cdots,G_{g}$ are simultaneously realizable on a $2^{\ell}\times k$ basis $M$ of rank 2
if and only if $\check{R}_{1},\check{R}_{2},\cdots,\check{R}_{r}$ and $\check{G}_{1},\check{G}_{2},\cdots,\check{G}_{g}$ are simultaneously realizable on a $2^{\ell}\times 2$ basis $M_{(2)}$ of rank 2.
\end{theorem}


\begin{proof} If there is a common basis $M$ such that $R_{i}=\underline{R}_{i}M^{\otimes n}, M^{\otimes n}G_{j}=\underline{G}_{j}$, then
\begin{equation*}
R_{i}'=R_{i}\displaystyle X'^{\otimes n}_{\sigma,\tau}=\underline{R}_{i}(MX'_{\sigma, \tau})^{\otimes n},
\end{equation*}
\begin{equation*}
(MX'_{\sigma, \tau})^{\otimes n} G_{j}'=
(MX'_{\sigma, \tau})^{\otimes n} (X'^{-1}_{\sigma,\tau})^{\otimes n}G_{j}=M^{\otimes n}G_{j}=\underline{G}_{j}.
\end{equation*}
From $M=(\alpha_{\sigma}, \alpha_{\tau})X_{\sigma, \tau}$, we have $MX'_{\sigma, \tau}=
\begin{pmatrix}
\alpha_{\sigma}&\alpha_{\tau}&0&\cdots&0
\end{pmatrix}$, then
\begin{center}
$\check{R}_{i}=\underline{R}_{i}(\alpha_{\sigma}~\alpha_{\tau})^{\otimes n}$,
~~~~ $(\alpha_{\sigma}~\alpha_{\tau})^{\otimes n}\check{G}_{j}=\underline{G}_{j}$.
\end{center}
Thus $\check{R}_{1},\check{R}_{2},\cdots,\check{R}_{r}$ and $\check{G}_{1}, \check{G}_{2}, \cdots$, $\check{G}_{g}$ are simultaneously realizable
on $(\alpha_{\sigma}~\alpha_{\tau})$.


Conversely,  if $\check{R}_{1},\check{R}_{2},\cdots,\check{R}_{r}$ and $\check{G}_{1},\check{G}_{2},\cdots,\check{G}_{g}$  are simultaneously realizable, then there is a $2^{\ell}\times 2$ basis $M_{(2)}=(\gamma_{1}~\gamma_{2})$
such that


\begin{center}
$\check{R}_{i}=\underline{R}_{i}M_{(2)}^{\otimes n}$, $M_{(2)}^{\otimes n}\check{G}_{j}=\underline{G}_{j}$.
\end{center}


Recall that $R_{1},R_2,\cdots,R_{r}$ are assumed to be simultaneously realizable.
Then by Theorem \ref{theorem3.2} and Lemma \ref{lemma3.3}
\[
\begin{split}
R_{i}&=R_{i}^{(\sigma, \tau)}\displaystyle X^{\otimes n}_{\sigma, \tau}=\check{R}_{i}\displaystyle X^{\otimes n}_{\sigma, \tau}
=\underline{R}_{i}(M_{(2)}X_{\sigma, \tau})^{\otimes n}.\end{split}\]
Furthermore,  since
\[
\begin{split}
&\begin{pmatrix}
\gamma_{1}&\gamma_{2}&0&\cdots&0
\end{pmatrix}\\
&=(\gamma_{1}~\gamma_{2})
\begin{pmatrix}
1&0&0&\cdots&0\\
0&1&0&\cdots&0
\end{pmatrix}\\
&=M_{(2)}X_{\sigma, \tau}X'_{\sigma, \tau}\end{split}\]
and
\[\begin{pmatrix}
1&0&0&\cdots&0\\
0&1&0&\cdots&0
\end{pmatrix}^{\otimes n}G'_{j}=\check{G}_{j},\]
we have
\[
\begin{split}
\underline{G}_{j}=M_{(2)}^{\otimes n}\check{G}_{j}=\begin{pmatrix}
\gamma_{1}&\gamma_{2}&0&\cdots&0
\end{pmatrix}^{\otimes n}G_{j}'=(M_{(2)}X_{\sigma, \tau})^{\otimes n}G_{j}.
\end{split}
\]
This implies that $R_{1},R_{2},\cdots,R_{r}$ and $G_{1},G_{2},\cdots,G_{g}$
 are simultaneously realizable over the basis  $M_{(2)}X_{\sigma, \tau}$.
\end{proof}

\begin{remark}
Note that $\check{R}_{1},\check{R}_{2},\cdots,\check{R}_{r}$ and $\check{G}_{1},\check{G}_{2},\cdots,\check{G}_{g}$ are signatures on domain size 2.
If all of the recognizers are degenerate, the holographic algorithm is trivial. Otherwise,
we can reduce SRP of $R_{1}, R_{2}, \cdots, R_{r}$ (on domain size $k\geq 3$) on a basis of rank 2 to SRP on domain size 2 by Theorem \ref{theorem3.2}.
Furthermore, if  $R_{1}, R_{2}, \cdots, R_{r}$ can be simultaneously realized on a basis of rank 2, we can reduce SRP on domain size $k\geq 3$ to SRP on domain size 2 by Theorem \ref{theorem3.3}.
\end{remark}

\begin{coro}\label{coro4.1}
The {\it contraction} of $\bigotimes_{j=1}^{g}\check{G}_{j}$ and $\bigotimes_{i=1}^{r}\check{R}_{i}$ is equal to the {\it contraction} of $\bigotimes_{j=1}^{g}G_{j}$ and $\bigotimes_{i=1}^{r}R_{i}$.
\end{coro}
\begin{proof}
 The {\it contraction} of $\bigotimes_{j=1}^{g}G'_{j}$ and $\bigotimes_{i=1}^{r}R'_{i}$ is equal to the {\it contraction} of $\bigotimes_{j=1}^{g}G_{j}$ and $\bigotimes_{i=1}^{r}R_{i}$ since $R'_{i}=R_{i}(X_{\sigma, \tau}')^{\otimes n}, G'_{j}=(X_{\sigma, \tau}'^{-1})^{\otimes n}G_{j}$.
Furthermore, the {\it contraction} of $\bigotimes_{j=1}^{g}\check{G}_{j}$ and $\bigotimes_{i=1}^{r}\check{R}_{i}$ is equal to the {\it contraction} of $\bigotimes_{j=1}^{g}G'_{j}$ and $\bigotimes_{i=1}^{r}R'_{i}$ since $R'_{i~i_{1}i_{2}\cdots i_{n}}=0$ if there exists $i_{j}\notin \{1, 2\}$ for $1\leq j\leq n$.
So the {\it contraction} of $\bigotimes_{j=1}^{g}\check{G}_{j}$ and $\bigotimes_{i=1}^{r}\check{R}_{i}$ is equal to the {\it contraction} of $\bigotimes_{j=1}^{g}G_{j}$ and $\bigotimes_{i=1}^{r}R_{i}$.
\end{proof}

\begin{coro}\label{coro4.2}
If
$R_{1},R_{2},\cdots,R_{r}$ and $G_{1},G_{2},\cdots,G_{g}$ are simultaneously realizable on a $2^{\ell}\times k$ basis of rank 2 and $R_{1}$ is non-degenerate,
then there exists a $2^{\ell}\times 2$ basis $M_{(2)}$ such that  $\check{R}_{1},\check{R}_{2},\cdots,\check{R}_{r}$ and $\check{G}_{1},\check{G}_{2},\cdots,\check{G}_{g}$ are simultaneously realizable on $M_{(2)}$ and $R_{1},R_{2},\cdots,R_{r}$ and $G_{1},G_{2},\cdots,G_{g}$ are simultaneously realizable on $M_{(2)}X_{\sigma, \tau}$.
\end{coro}

\section{The Collapse Theorem}

\begin{theorem}\cite{string6}(The Collapse Theorem on Domain Size 2)\label{collapse for domain 2}
Let $R_{1}, R_{2}, \cdots, R_{r}$ and $G_{1}, G_{2}, \cdots, G_{g}$ be recognizers and generators respectively that a holographic algorithm on domain size 2 employs, where there is at least one generator that is non-degenerate. If there exists a $2^{\ell}\times 2$ basis $M$ of rank 2 such that $R_{1}, R_{2}, \cdots, R_{r}$ and $G_{1}, G_{2}, \cdots, G_{g}$ are simultaneously realizable, then there exists a $2\times 2$ basis $M'$ of rank 2 such that $R_{1}, R_{2}, \cdots, R_{r}$ and $G_{1}, G_{2}, \cdots, G_{g}$ are simultaneously realizable.
\end{theorem}



\begin{theorem}(The Collapse Theorem on Domain Size $k \ge 3$ by Basis of Rank 2)
Let $R_{1}, R_{2}, \cdots, R_{r}$ and $G_{1}, G_{2}, \cdots, G_{g}$ be recognizers and generators respectively that a holographic algorithm on domain size $k$ employs.  Assume that there exists at least one recognizer is non-degenerate. Otherwise, the holographic algorithm is trivial. If there exists a $2^{\ell}\times k$ basis $M$ of rank 2 such that $R_{1}, R_{2}, \cdots, R_{r}$ and $G_{1}, G_{2}, \cdots, G_{g}$ are simultaneously realizable, then we can get $\check{R}_{1}, \check{R}_{2}, \cdots, \check{R}_{r}$ and $\check{G}_{1}, \check{G}_{2}, \cdots, \check{G}_{g}$ (Following the notations of Theorem \ref{theorem3.3}). Then either all of $\check{G}_{1}, \check{G}_{2}, \cdots, \check{G}_{g}$ are degenerate, in this case the holographic algorithm is trivial from Corollary \ref{coro4.1}, or
there exists a $2\times k$ basis $M'$ of rank 2 such that $R_{1}, R_{2}, \cdots, R_{r}$ and $G_{1}, G_{2}, \cdots, G_{g}$ are simultaneously realizable.
\end{theorem}
\begin{proof}
 Assume that $R_{1}$ is non-degenerate, then we have $\check{R}_{1},\check{R}_{2},\cdots,\check{R}_{r}$ and $\check{G}_{1},\check{G}_{2},\cdots,\check{G}_{g}$, which are signatures on domain size 2 and can be simultaneously realized on a $2^{\ell}\times 2$ basis from Theorem \ref{theorem3.3}. 
Since $\check{G}_{1}, \check{G}_{2}, \cdots, \check{G}_{g}$ are not all degenerate,
from Theorem \ref{collapse for domain 2},
$\check{R}_{1}, \check{R}_{2}, \cdots, \check{R}_{r}$ and $\check{G}_{1}, \check{G}_{2}, \cdots, \check{G}_{g}$ can be realized on a $2\times 2$ basis $M'_{(2)}$. Then from
 Corollary \ref{coro4.2}
$R_{1}, R_{2}, \cdots, R_{r}$ and $G_{1}, G_{2}, \cdots, G_{g}$ can be simultaneously realized on the $2\times k$ basis $M'_{(2)}X_{\sigma, \tau}$.
\end{proof}



\bibliographystyle{model1-num-names}







\end{document}